\documentclass[11pt]{article}
\usepackage{latexsym}
\usepackage{amssymb}
\usepackage{bussproofs}
\usepackage{pstricks}
\usepackage{color}

\newtheorem{theorem}{Theorem}

\newtheorem{lemma}[theorem]{Lemma}

\newenvironment{definition}{\begin{trivlist}\item[]{\bf Definition}\ }%
{\end{trivlist}}
\newenvironment{proof}{\begin{trivlist}\item[]{\bf Proof\ }}%
{\end{trivlist}}

\def\lorunary{{\mathop{\bigvee}}}           
\def\falseBPO#1{{\left(\lorunary{#1}\right)}}
\def\falseBPOpi{{\falseBPO{\overline\pi}}}
\def\falseBPOpisub#1{{\falseBPO{\overline\pi_{#1}}}}

\newcommand{\qed}{\Box}			

\newcommand{\rddots}{{\mathinner{\mkern1mu\raise1pt\vbox{\kern7pt\hbox{.}}%
        \mkern2mu\raise4pt\hbox{.}\mkern2mu\raise7pt\hbox{.}\mkern1mu}}}
\newcommand{\proofdots}{{\ddots\vdots\,\rddots}}

\def\GT{{\mbox{\textup{GT}}}}
\def\GGT{{\mbox{\textup{GGT}}}}

\def\caseiti{{(\kern-1pt{\it i})}}
\def\caseitii{{(\kern-1pt{\it ii})}}
\def\caseitiii{{(\kern-1pt{\it iii})}}
\def\caseitiv{{(\kern-1pt{\it iv})}}
\def\caseitip{{(\kern-1pt{\it i'})}}
\def\caseitiip{{(\kern-1pt{\it ii'})}}
\def\caseitiiip{{(\kern-1pt{\it iii'})}}
\def\caseitivp{{(\kern-1pt{\it iv'})}}

\begin{document}

\title{An Improved Separation of Regular Resolution
from Pool Resolution and Clause Learning \\
{\em \normalsize Preliminary version. Comments appreciated.}}

\author{
Maria Luisa Bonet\thanks{Supported in part by grant TIN2010-20967-C04-02.}\\
\small Lenguajes y Sistemas Inform{\'a}ticos \\
\small Universidad Polit{\'e}cnica de Catalu{\~n}a \\
\small Barcelona, Spain \\
\small \tt bonet@lsi.upc.edu
\and
Sam Buss\thanks{Supported in part by
NSF grants \hbox{DMS-0700533} and \hbox{DMS-1101228}, and by a grant
from the Simons
Foundation (\#208717 to Sam Buss).
The second author thanks the John Templeton Foundation for supporting
his participation in the CRM Infinity Project
at the Centre de Recerca Matem\`atica,
Barcelona, Catalonia, Spain during which
these results were obtained.}\\
\small Department of Mathematics \\
\small University of California, San Diego\\
\small La Jolla, CA 92093-0112, USA\\
\small \tt sbuss@math.ucsd.edu
}

\maketitle

\begin{abstract}
We prove that the graph tautology principles of
Alekhnovich, Johannsen, Pitassi and Urquhart
have polynomial size pool resolution refutations that use
only input lemmas as learned clauses and without degenerate
resolution inferences.  We also prove that these
graph tautology principles can be refuted
by polynomial size DPLL proofs with
clause learning, even when restricted to greedy,
unit-propagating DPLL
search.
\end{abstract}

\section{Introduction}

The problem SAT of deciding the satisfiability of propositional CNF formulas
is of great theoretical and practical interest. Even though it is
NP-complete, industrial instances with hundreds of thousands variables
are routinely solved
by state of the art SAT solvers.  Most of these solvers are
based on the DPLL
procedure extended with clause learning, restarts, variable
selection heuristics, and other techniques.

The basic DPLL procedure
without clause learning is equivalent to tree-like
resolution.  The addition of clause learning makes DPLL considerably
stronger.  In fact, clause learning together with unlimited
restarts is capable of simulating general resolution
proofs~\cite{PipatsrisawatDarwiche:clauselearning}.
However, the
exact power of DPLL with clause learning but without restarts
is unknown.  This question
is interesting not only for theoretical
reasons, but also because of the potential for better understanding
the practical performance of
various refinements of DPLL with clause learning.

Beame, Kautz, and Sabharwal~\cite{BKS:clauselearning} gave
the first
theoretical analysis of DPLL
with clause learning.
Among other things,
they noted that clause learning with restarts simulates
general resolution.  Their construction required the
DPLL algorithm to ignore some contradictions,
but this situation was rectified
by Pipatsrisawat and
Darwiche~\cite{PipatsrisawatDarwiche:clauselearning} who showed
that SAT solvers which do not ignore contradictions can also
simulate resolution.   These techniques were also applied
to learning bounded width clauses by~\cite{AFT:clauseLearning}.

Beame et al.~\cite{BKS:clauselearning} also studied DPLL clause
learning without restarts.  Using a method of ``proof trace extensions'',
they were able to show that DPLL with clause learning
and no restarts is strictly
stronger than any ``natural'' proof system strictly weaker than
resolution.  Here, a {\em natural} proof system is one in which
proofs do not increase in length when variables are restricted
to constants.  The class of natural proof systems is known to
include common proof systems such as tree-like or regular proofs.
The proof trace method involves introducing extraneous variables
and clauses, which have the effect of giving
the clause learning DPLL algorithm more freedom in choosing
decision variables for branching.

There have been two approaches to formalizing DPLL with clause
learning as a static proof system rather
than as a proof search algorithm.  The first is
pool resolution with a degenerate resolution inference, due
originally to Van Gelder~\cite{VanGelder:PoolResolution} and
studied further by Bacchus et al.~\cite{HBPvG:clauselearn}.
Pool resolution requires proofs to have a depth-first regular
traversal similarly to the search space of a DPLL algorithm.
Degenerate resolution allows resolution inferences in which one or both of
the hypotheses may be lacking occurrences of the resolution
literal.  Van Gelder argued that pool resolution with degenerate
resolution inferences simulates a wide range of DPLL algorithms
with clause learning.  He also gave a proof,
based on~\cite{AJPU:regularresolution}, that pool
resolution with degenerate inferences is stronger than
regular resolution, using extraneous
variables similar to proof trace extensions.

The second approach is due to
Buss-Hoffmann-Johannsen~\cite{BHJ:ResTreeLearning}
who introduced a ``partially degenerate'' resolution
rule called w-resolution, and a proof system regWRTI
based on w-resolution and clause learning of
``input lemmas''.  They proved that
regWRTI exactly captures non-greedy DPLL with clause learning.
By ``non-greedy'' is meant that contradictions may need to be
ignored by the DPLL search.

Both \cite{HBPvG:clauselearn} and~\cite{BHJ:ResTreeLearning}
gave improved versions of the proof trace extension method
so that the extraneous variables depend only on the
set of clauses being refuted and not on resolution
refutation of the clauses.
The drawback remains, however, that the proof trace extension
method gives contrived sets of clauses and contrived resolution
refutations.

It remains open whether any of
DPLL with clause learning,
pool resolution (with or without
degenerate inferences),
or the regWRTI proof system can polynomially
simulate general resolution.
One approach to answering
these questions is to try to separate pool resolution (say)
from general resolution.
So far, however, separation results are
known only for the weaker system of regular
resolution, based on work of
Alekhnovitch et al.~\cite{AJPU:regularresolution},
who gave an exponential separation between
regular resolution and general resolution.
Alekhnovitch et al.~\cite{AJPU:regularresolution}
proved this separation for two families
of tautologies, variants of the graph tautologies~$\GT^\prime$ and
the ``Stone'' pebbling tautologies.
Urquhart~\cite{Urquhart:regularresolution} subsequently
gave a related separation.\footnote{%
Huang and Yu~\cite{HuangYu:regularresolution} also gave a
separation of regular resolution and general resolution, but
only for a single set of clauses.
Goerdt~\cite{Goerdt:regularresolution}
gave a quasipolynomial separation of regular resolution and
general resolution.}
In the present paper, we call the tautologies~$\GT^\prime$
the {\em guarded} graph tautologies, and henceforth denote
them $\GGT$ instead of $\GT^\prime$; their definition is given
in Section~\ref{PrelimsSect}.

Thus, an obvious question is whether
pool resolution (say) has polynomial size proofs of
the $\GGT$ tautologies or the Stone tautologies.
The main result of the present paper resolves the first question
by showing that pool resolution does indeed have polynomial
size proofs of the graph tautologies~$\GGT$.
Our proofs apply to the original $\GGT$ principles,
without the use of extraneous variables in the style
of proof trace extensions;
our refutations use only the traditional
resolution rule and do not require
degenerate resolution inferences or w-resolution
inferences.   In addition, we use only learning of input clauses;
thus, our refutations are also regWRTI
proofs (and in fact regRTI proofs) in the terminology
of~\cite{BHJ:ResTreeLearning}.
As a corollary of the characterization of regWRTI
by~\cite{BHJ:ResTreeLearning},
the $\GGT$ principles
have polynomial size refutations that can be found by
a DPLL algorithm with clause learning and without restarts (under
the appropriate variable selection order).

It is still open if there are polynomial size pool resolution
refutations for the Stone principles.  However, it is plausible
that our methods could extend to give such refutations.
It seems more likely that our proof methods could extend to the
pebbling tautologies used by~\cite{Urquhart:regularresolution},
as the hardness of those tautologies is due to the addition
of randomly chosen ``guard'' literals, similarly to the
$\GGT$ tautologies.\footnote{Subsequent to the circulation of
a preliminary version of the present paper, Buss and Johanssen
[in preparation]
have succeeded giving polynomial size regRTI proofs of the
pebbling tautologies of~\cite{Urquhart:regularresolution}.}
A much more ambitious project
would be to show that pool resolution or regWRTI can simulate
general resolution, or that DPLL with clause learning and without
restarts can simulate general resolution.  It is far from clear
that this is true, but, if so, our methods below may
represent a step in that direction.

The outline of the paper is as follows.  Section~\ref{PrelimsSect}
begins with the definitions of resolution, degenerate resolution, and
w-resolution, and then regular, tree, and pool resolution.
After that, we define the
graph tautologies $\GT_n$ and the guarded versions $\GGT_n$,
and state the main theorems about proofs
of the $\GGT_n$ principles.  Section~\ref{MainPfSect} gives
the proof of the theorems about pool resolution
and regRTI proofs.  Several ingredients are
needed for the proof.
The first idea is to try to
follow the regular refutations of the
graph tautology clauses $\GT_n$ as given by
St{\aa}lmarck~\cite{Stalmarck:trickyformulas}
and Bonet and Galesi~\cite{BonetGalesi:ResAndPolyCalculus}:
however, these refutations
cannot be used directly since the transitivity clauses
of $\GT_n$ are ``guarded'' in the $\GGT_n$ clauses
and this yields refutations which
violate the regularity/pool property.  So, the
second idea is that
the proof search process branches as needed to learn
transitivity clauses.  This generates additional clauses
that must be proved: to handle these, we develop a notion
of ``bipartite partial order'' and show that the
refutations of
\cite{Stalmarck:trickyformulas,BonetGalesi:ResAndPolyCalculus}
can still be used in the presence of a bipartite partial order.
The tricky part is to be sure that exactly the right set of
clauses is derived by each subproof.
Some straightforward bookkeeping
shows that the resulting proof is polynomial size.

Section~\ref{GreedySect} discusses how to modify
the refutations constructed for Section~\ref{MainPfSect}
so that they are ``greedy'' and ``unit-propagating.
These conditions means that proofs cannot ignore contradictions,
nor contradictions that can be obtained by unit propagation.
The
greedy and unit-propagating conditions correspond well to actual
implemented DPLL proof search algorithms, since they backtrack whenever
a contradiction can be found by unit propagation.  Section~\ref{GreedySect}
concludes with an explicit description of a polynomial time
DPLL clause learning
algorithm for the $\GGT_n$ clauses.

We are grateful to J.\ Hoffmann and J.\ Johannsen for a correction
to an earlier version of the proof of Theorem~\ref{regRtiGgtThm}.
We also thank A.\ Van Gelder, A.\ Beckmann, and T.\ Pitassi for
encouragement, suggestions, and useful comments.

\section{Preliminaries and main results}\label{PrelimsSect}

Propositional formulas are defined over a
set of variables and the connectives $\wedge$, $\vee$ and $\neg$.
We use the notation $\overline{x}$
to express the negation~$\neg x$ of~$x$. A {\em literal} is either
a variable~$x$ or a negated variable~$\overline{x}$.
A {\em clause}~$C$ is a set of literals, interpreted
as the disjunction of its members.
The empty clause,~$\Box$, has truth value {\em False}.
We shall only use formulas
in {\em conjunctive normal form}, CNF; namely, a formula will be
a set (conjunction) of clauses.  We often use disjunction~($\lor$)
and union~($\cup$) interchangeably.

\begin{definition}
The various forms of resolution
take two clauses $A$ and $B$ called the {\em premises}
and a literal $x$
called the {\em resolution variable},
and produce a new clause $C$ called the {\em resolvent}.
\begin{prooftree}
\AxiomC{A} \AxiomC{B}
\BinaryInfC{C}
\end{prooftree}
In all cases below, it is required that $\overline x \notin A$ and
$x\notin B$.
The different forms of resolution are:
\begin{description}
\item {\em Resolution rule.}  The hypotheses have the forms
$A := A'\lor x$ and $B:=B'\lor \, \overline{x}$.
The resolvent~$C$ is $A'\lor B'$.
\item {\em Degenerate resolution rule.} \cite{HBPvG:clauselearn,VanGelder:PoolResolution}
If $x\in A$ and $\overline{x}\in B$,
we apply the resolution rule to obtain~$C$.
If $A$ contains~$x$, and $B$ doesn't contain~$\overline x$,
then the resolvent $C$ is $B$.
If $A$ doesn't contain~$x$,
and $B$ contains~$\overline{x}$,
then the resolvent~$C$ is~$A$.
If neither $A$ nor~$B$ contains the literal $x$ or~$\overline x$,
then $C$ is the lesser of $A$ or~$B$ according to some tiebreaking
ordering of clauses.
\item{\em w-resolution rule.} \cite{BHJ:ResTreeLearning}
From $A$ and~$B$ as above, we infer
$C:=(A\setminus\{x\})\lor(B\setminus\{\overline{x}\})$.
If the literal $x\notin A$ (resp., $\overline x\notin B$), then it
is called a {\em phantom literal} of $A$ (resp.,~$B$).
\end{description}
\end{definition}

\begin{definition}
A {\em resolution derivation}, or {\em proof}, of a clause~$C$
from a CNF formula~$F$
is a sequence of clauses $C_1,\ldots,C_s$ such that $C=C_s$
and such that each clause
from the sequence is either a clause from~$F$ or
is the resolvent of two previous clauses.
If the derived clause,~$C_s$, is the empty clause, this is
called a {\em resolution refutation} of~$F$.
The more general systems of
degenerate and w-resolution refutations are defined similarly.
\end{definition}

We can represent a derivation
as a directed acyclic graph (dag) on the
vertices $C_1,\ldots,C_s$,
where each clause from~$F$ has out-degree~$0$,
and all the other vertices from $C_1,\ldots,C_s$ have edges pointing
to the two clauses from which they were derived.
The empty clause has in-degree~$0$.
We use the terms ``proof'' and ``derivation'' interchangeably.

Resolution is sound and complete
in the refutational sense: a CNF
formula~$F$ has a refutation if
and only if $F$ is unsatisfiable, that is, if and only
if $\lnot F$ is a tautology.
Furthermore, if there is a derivation of a clause~$C$
from~$F$,
then $C$ is a consequence of~$F$;
that is, for every truth assignment~$\sigma$,
if $\sigma$ satisfies $F$ then it satisfies~$C$.
Conversely, if $C$ is a consequence of $F$ then
there is a derivation
of some $C'\subseteq C$ from~$F$.

A resolution refutation is {\em regular} provided that,
along any path in the directed acyclic graph,
each variable is resolved at most once.  A resolution
derivation of a clause~$C$ is {\em regular} provided
that, in addition, no variable appearing in~$C$ is used as
a resolution variable in the derivation.
A refutation is {\em tree-like} if the underlying graph is a tree;
that is, each occurrence of a clause occurring in the refutation
is used at most once as a
premise of an inference.

We next define a version of pool resolution,
using the conventions of~\cite{BHJ:ResTreeLearning} who called
this ``tree-like regular resolution with lemmas''.
The idea is that clauses obtained previously in the proof
can be used freely as learned lemmas.
To be able to talk about clauses previously obtained,
we need to define an ordering of clauses.

\begin{definition}
Given a tree $T$, the {\em postorder}
ordering $<_T$ of the nodes is defined as follows:
if $u$ is a node of~$T$,
$v$~is a node in the subtree rooted at the left child of~$u$,
and $w$~is a node in the subtree rooted at the right child
of~$u$, then $v<_T w<_T u$.
\end{definition}

\begin{definition}
A {\em pool resolution} proof from a set of initial clauses~$F$
is a resolution proof tree~$T$
that fulfills the following conditions:
(a)~each leaf is labeled with either a clause of~$F$ or a clause
(called a ``lemma'')
that appears earlier in the tree in the $<_T$ ordering;
(b)~each internal node is labeled with a clause and a literal,
and the clause is obtained by resolution
from the clauses labeling the node's children
by resolving on the given literal;
(c)~the proof tree is regular;
(d)~the roof is labeled with the conclusion clause.
If the labeling of the root is the empty
clause $\Box$, the pool resolution proof
is a {pool refutation}.
\end{definition}

The notions of {\em degenerate pool resolution} proof and
{\em pool w-resolution} proof are
defined similarly, but allowing degenerate resolution or w-resolution
inferences, respectively.
Note that the two papers
\cite{VanGelder:PoolResolution,HBPvG:clauselearn} defined pool
resolution to be the degenerate pool resolution system, so our notion
of pool resolution is more restrictive than theirs.  (Our definition
is equivalent to the one in~\cite{Buss:poolhard}, however.)

A ``lemma'' in part~(a) of the above definition
is called an {\em input lemma} if it is derived by {\em input}
subderivation, namely by a subderivation
in which each inference has at least one
hypothesis which is a member of~$F$ or is a lemma.

Next we define various graph tautologies, sometimes also
called ``ordering principles''.  They will all
use a size parameter~$n>1$, and variables $x_{i,j}$ with $i,j\in[n]$
and $i\not= j$, where $[n] = \{0,1,2,\ldots,n{-}1\}$.   A variable~$x_{i,j}$
will intuitively represent the condition that $i\prec j$ with $\prec$
intended to be a total, linear order.  We will thus
always adopt the simplifying
convention that $x_{i,j}$ and $\overline x_{j,i}$ are
the identical literal.
This identification makes no essential difference to
the complexity of proofs of the tautologies,
but it reduces the number of literals and clauses,
and simplifies the definitions.

The following principle is based on the tautologies defined by
Krishnamurthy \cite{Krishnamurthy:trickyformulas}.
These tautologies, or similar ones,
have also been studied by \cite{Stalmarck:trickyformulas,%
BonetGalesi:ResAndPolyCalculus,%
AJPU:regularresolution,%
BeckmannBuss:dLK,%
SBI:SwitchingkDnf,%
VanGelder:InputCoverNumber}.

\begin{definition}
Let $n>1$.  Then $\GT_n$ is the following set of
clauses involving the variables $x_{i,j}$, for $i,j\in [n]$ with $i\not= j$.
\begin{enumerate}
\item[($\alpha_\emptyset$)]
The clauses $\bigvee_{j \not= i} x_{j,i}$, for each
value $i<n$.
\item[($\gamma_\emptyset$)] The {\em transitivity clauses} $T_{i,j,k}:=
\overline x_{i,j} \lor \overline x_{j,k} \lor \overline x_{k,i}$
for all distinct $i,j,k$ in $[n]$.
\end{enumerate}
\end{definition}

Note that the clauses $T_{i,j,k}$, $T_{j,k,i}$ and $T_{k,i,j}$ are identical.
For this reason Van Gelder \cite{VanGelder:PoolResolution}
uses the name "no triangles" (NT) for a similar principle.

The next definition is from~\cite{AJPU:regularresolution}, who
used the notation $\GT^\prime_n$.  They
used particular functions $r$ and~$s$ for their lower bound proof,
but since our upper bound proof does not depend on
the details of $r$ and~$s$
we leave them unspecified. We require that $r(i,j,k)\not=s(i,j,k)$ and that
the set $\{r(i,j,k),s(i,j,k)\}\not\subset\{i,j,k\}$.
In addition, w.l.o.g.,
$r(i,j,k)=r(j,k,i)=r(k,i,j)$, and similarly for~$s$.

\begin{definition}  Let $n\ge 1$, and let $r(i,j,k)$ and $s(i,j,k)$ be
functions mapping $[n]^3 \rightarrow [n]$ as above.  The {\em guarded
graph tautology} $\GGT_n$ consists of the following
clauses:
\begin{enumerate}
\item[($\alpha_\emptyset$)]
The clauses $\bigvee_{j \not= i} x_{j,i}$, for each
value $i<n$.
\item[($\gamma^\prime_\emptyset$)] The {\em guarded} transitivity clauses
$T_{i,j,k}\lor x_{r,s}$
and $T_{i,j,k}\lor \overline x_{r,s}$,
for all distinct $i,j,k$ in $[n]$, where $r=r(i,j,k)$ and $s=s(i,j,k)$.
\end{enumerate}
\end{definition}

Our main result is:
\begin{theorem}\label{PoolResGgtThm}
The guarded graph tautology principles $\GGT_n$ have
polynomial size pool resolution refutations.
\end{theorem}

The proof of Theorem~\ref{PoolResGgtThm} will construct
pool refutations in the form of regular tree-like refutations
with lemmas.
A key part of this is
learning transitive closure clauses
that are derived using resolution
on the guarded transitivity clauses of~$\GGT_n$.
A slightly modified construction, that uses a
result from~\cite{BHJ:ResTreeLearning},
gives instead tree-like regular resolution
refutations with {\em input} lemmas.  This will establish
the following:
\begin{theorem}\label{regRtiGgtThm}
The guarded graph tautology principles $\GGT_n$ have
polynomial size, tree-like regular resolution refutations with input lemmas.
\end{theorem}
A consequence of Theorem~\ref{regRtiGgtThm} is that
the $\GGT_n$ clauses can be shown unsatisfiable by
non-greedy polynomial size DPLL searches using clause learning.
This follows via
Theorem~5.6 of~\cite{BHJ:ResTreeLearning},
since the refutations of~$\GGT_n$ are regRTI, and hence regWRTI, proofs
in the sense of~\cite{BHJ:ResTreeLearning}.

However, as discussed in Section~\ref{GreedySect}, we can improve
the constructions of Theorems \ref{PoolResGgtThm} and~\ref{regRtiGgtThm}
to show that
the $\GGT_n$ principles can be refuted also by
{\em greedy} and {\em unit-propagating} polynomial size DPLL searches with
clause learning.

\section{Proof of main theorem}\label{MainPfSect}

The following theorem is an important ingredient of our upper bound proof.

\begin{theorem}\label{GtProofsThm}
{\rm (St{\aa}lmarck~\cite{Stalmarck:trickyformulas};
Bonet-Galesi~\cite{BonetGalesi:ResAndPolyCalculus};
Van Gelder~\cite{VanGelder:InputCoverNumber})}
The sets $\GT_n$ have regular resolution
refutations~$P_n$ of polynomial size~$O(n^3)$.
\end{theorem}

We do not include a direct proof of Theorem~\ref{GtProofsThm}
here, which can be found in
\cite{Stalmarck:trickyformulas,BonetGalesi:ResAndPolyCalculus,VanGelder:InputCoverNumber}.
The present paper uses the proofs~$P_n$ as a ``black box'';
the only property needed is that the $P_n$'s are
regular and polynomial size.
Lemma~\ref{BpoDerivationLm} below
is a direct generalization to Theorem~\ref{GtProofsThm}; in fact,
when specialized to
the case of $\pi =\emptyset$, it is identical to Theorem~\ref{GtProofsThm}.

The refutations~$P_n$ can be modified to give refutations
of $\GGT_n$ by first deriving each transitive clause~$T_{i,j,k}$
from the two guarded transitivity clauses of~$(\gamma^\prime_\emptyset)$.
This however destroys the regularity property, and
in fact no polynomial size regular refutations exist
for $\GGT_n$~\cite{AJPU:regularresolution}.

As usual, a {\em partial order} on $[n]$ is an
antisymmetric, transitive relation binary relation
on~$[n]$.  We shall be mostly interested in ``partial
specifications'' of partial orders: partial
specifications are not
required to be transitive.

\begin{definition} A {\em partial specification},~$\tau$, of a
partial order is a set of ordered pairs $\tau\subseteq[n]\times[n]$
which are consistent with some (partial) order.
The minimal partial order containing~$\tau$ is the
transitive closure of~$\tau$.  We
write $i\prec_\tau j$ to denote $\langle i,j\rangle \in \tau$,
and write $i\prec_\tau^* j$ to denote that $\langle i,j\rangle$ is in the
transitive closure of~$\tau$.

The {\em $\tau$-minimal} elements are the $i$'s such that
$j\prec_\tau i$ does not hold for any~$j$.
\end{definition}

We will be primarily interested in particular
kinds of partial orders, called ``bipartite'' partial orders,
that can be associated with partial orders.  A bipartite partial order is a
partial order that does not have any chain of inequalities
$x\prec y\prec z$.

\begin{definition}
A {\em bipartite partial order} is a binary relation~$\pi$
on $[n]$
such that the domain and range of~$\pi$
do not intersect.
The set of $\pi$-minimal elements is denoted $M_\pi$.
\end{definition}
The righthand side of Figure~\ref{BiPartiteExampleFig} shows an example.
The bipartiteness of~$\pi$ arises from the fact that $M_\pi$ and
$[n]\setminus M_\pi$ partition $[n]$ into two sets.
Note that if $i \prec_\pi j$, then $i\in M_\pi$ and $j\notin M_\pi$.
In addition, $M_\pi$~contains the isolated points of~$\pi$.
\begin{definition}
Let $\tau$ be a specification of a partial order.
The bipartite partial order~$\pi$ that is {\em associated with}
$\tau$ is defined by letting $i\prec_\pi j$ hold for precisely
those $i$ and~$j$ such that $i$ is $\tau$-minimal
and $ i \prec^*_\tau j$.
\end{definition}
It is easy to check that the $\pi$ associated with~$\tau$
is in fact a bipartite partial order.
The intuition is that $\pi$~retains only the information about
whether $i\prec^*_\tau j$ for minimal elements~$i$,
and forgets the ordering that $\tau$ imposes on
non-minimal elements. Figure~\ref{BiPartiteExampleFig} shows an example of
how to obtain a bipartite partial order from a partial specification.

\begin{figure}[t]
\begin{center}
\psset{unit=10mm,arrowscale=1.4 1.2}     
\begin{pspicture}(0,0)(4.7,2)
\pscircle*(0,0){2pt}
\uput[0](0,0){$1$}
\pscircle*(1,0){2pt}
\uput[0](1,0){$2$}
\pscircle*(2,0){2pt}
\uput[0](2,0){$3$}
\pscircle*(3,0){2pt}
\uput[0](3,0){$4$}
\pscircle*(4,0){2pt}
\uput[0](4,0){$5$}
\pscircle*(1.5,1){2pt}
\uput[0](1.5,1){$6$}
\pscircle*(2.5,1){2pt}
\uput[0](2.5,1){$7$}
\pscircle*(3.5,1){2pt}
\uput[0](3.5,1){$8$}
\pscircle*(4.5,1){2pt}
\uput[0](4.5,1){$9$}
\pscircle*(2,2){2pt}
\uput[0](2,2){$10$}
\pscircle*(4,2){2pt}
\uput[0](4,2){$11$}
\psline{->}(2,0)(1.5,1)
\psline{->}(1.5,1)(2,2)
\psline{->}(2.5,1)(2,2)
\psline{->}(3,0)(2.5,1)
\psline{->}(3,0)(3.5,1)
\psline{->}(4,0)(3.5,1)
\psline{->}(4,0)(4.5,1)
\psline{->}(4.5,1)(4,2)
\end{pspicture}
\hfill
\raise 10mm \hbox{$\Rightarrow$}
\hfill
\raise 5mm \hbox{
\begin{pspicture}(-1.5,0)(4.5,1)
\pscircle*(0,0){2pt}
\uput[0](0,0){$1$}
\pscircle*(1,0){2pt}
\uput[0](1,0){$2$}
\pscircle*(2,0){2pt}
\uput[0](2,0){$3$}
\pscircle*(3,0){2pt}
\uput[0](3,0){$4$}
\pscircle*(4,0){2pt}
\uput[0](4,0){$5$}
\pscircle*(1.5,1){2pt}
\uput[45](1.5,1){$6$}
\pscircle*(3,1){2pt}
\uput[90](3,1){$7$}
\pscircle*(3.5,1){2pt}
\uput[90](3.5,1){$8$}
\pscircle*(4,1){2pt}
\uput[90](4,1){$9$}
\pscircle*(2.5,1){2pt}
\uput[90](2.5,1){$10$}
\pscircle*(4.5,1){2pt}
\uput[90](4.5,1){$11$}
\rput(-1,1){$[n]-M_\pi$:}
\rput(-1,0){$M_\pi$:}
\psline{->}(2,0)(1.5,1)
\psline{->}(2,0)(2.5,1)
\psline{->}(3,0)(3,1)
\psline{->}(3,0)(2.5,1)
\psline{->}(3,0,1)(3,1)
\psline{->}(3,0)(3.5,1)
\psline{->}(4,0)(3.5,1)
\psline{->}(4,0)(4,1)
\psline{->}(4,0)(4.5,1)
\end{pspicture}
}
\end{center}
\caption{Example of a partial specification of a partial order (left)
and the associated bipartite partial order (right).}
\label{BiPartiteExampleFig}
\end{figure}
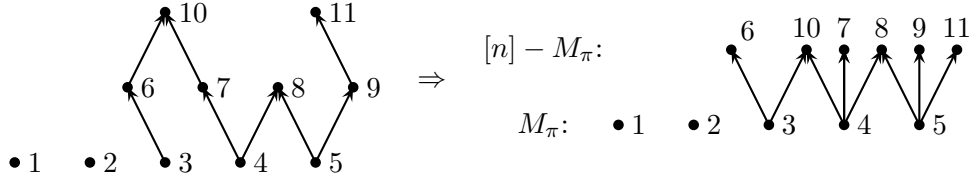

We define the graph tautology~$\GT_{\pi,n}$ relative to~$\pi$
as follows.
\begin{definition} Let $\pi$ be a bipartite partial order on~$[n]$.
Then $\GT_{\pi,n}$ is the set of clauses containing:
\begin{enumerate}
\item[($\alpha$)] The clauses $\bigvee_{j \not= i} x_{j,i}$, for each
value $i\in M_\pi$.
\item[($\beta$)] The transitivity clauses $T_{i,j,k}:=
\overline x_{i,j} \lor \overline x_{j,k} \lor \overline x_{k,i}$
for all distinct $i,j,k$ in $M_\pi$. (Vertices $i,j,k^\prime$
in Figure~\ref{BipartiteFig} show an example.)
\item[($\gamma$)] The transitivity clauses $T_{i,j,k}$
for all distinct $i,j,k$ such that $i,j\in M_\pi$ and
$i\not\prec_\pi k$ and $j\prec_\pi k$.
(As shown in
Figure~\ref{BipartiteFig}.)
\end{enumerate}
\end{definition}

The set $\GT_{\pi,n}$ is satisfiable if $\pi$ is nonempty.
As an example, there is the assignment that sets $x_{j,i}$ true
for some fixed $j\notin M_\pi$ and every $i\in M_\pi$, and sets
all other variables false.
However, if $\pi$
is applied as a restriction,
then $\GT_{\pi,n}$ becomes unsatisfiable.
That is to say, there is no assignment which
satisfies $\GT_{\pi,n}$ and is consistent with~$\pi$.
This fact
is proved by the regular derivation~$P_\pi$ described in the next
lemma.

\begin{definition} For $\pi$ a bipartite partial order, the clause
$\falseBPOpi$ is defined by
\[
\falseBPOpi ~:=~ \{ \overline x_{i,j} : i\prec_\pi j \},
\]
\end{definition}

\begin{lemma}\label{BpoDerivationLm}
Let $\pi$ be a bipartite partial order on~$[n]$.
Then there
is a regular derivation~$P_\pi$ of $\falseBPOpi$
from the set $\GT_{\pi,n}$.

The only variables resolved on in~$P_\pi$ are the following:
the variables $x_{i,j}$ such that $i,j\in M_\pi$,
and the variables $x_{i,k}$ such that $k\notin M_\pi$,
$i\in M_\pi$, and
$i\not\prec_\pi k$.
\end{lemma}

Lemma~\ref{BpoDerivationLm} implies
that if $\pi$ is the bipartite partial order
associated with a partial specification~$\tau$ of
a partial order, then the derivation~$P_\pi$ does
not resolve on any literal whose value is set
by~$\tau$.  This is proved by noting that if $i\prec_\tau j$,
then $j \notin M_\pi$.

Note that if $\pi$ is empty,
$M_\pi=[n]$ and there are no clauses
of type~($\gamma$).
In this case, $\GT_{\pi,n}$ is identical to~$\GT_n$,
and $P_\pi$ is the same as the refutation of~$\GT_n$
of Theorem~\ref{GtProofsThm}.

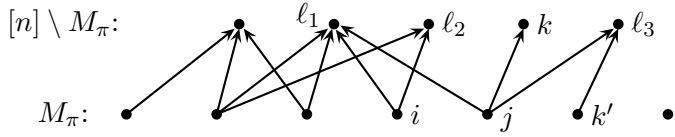
\begin{figure}[t]
\begin{center}
\psset{unit=0.06cm}     
\begin{pspicture}(-30,0)(140,20)
\pscircle*(0,0){2pt}
\pscircle*(20,0){2pt}
\pscircle*(40,0){2pt}
\pscircle*(60,0){2pt}
\pscircle*(80,0){2pt}
\pscircle*(100,0){2pt}
\pscircle*(120,0){2pt}
\rput(-14,0){$M_\pi$:}
\rput(-14,20){$[n]\setminus M_\pi$:}
\pscircle*(25,20){2pt}
\pscircle*(46,20){2pt}
\pscircle*(67,20){2pt}
\pscircle*(88,20){2pt}
\pscircle*(109,20){2pt}
\psline[arrowscale=1.4 1.2]{->}(0,0)(24,19)
\psline[arrowscale=1.4 1.2]{->}(20,0)(25,19)
\psline[arrowscale=1.4 1.2]{->}(40,0)(26,19)
\psline[arrowscale=1.4 1.2]{->}(20,0)(44.7,19)
\psline[arrowscale=1.4 1.2]{->}(40,0)(45.6,19)
\psline[arrowscale=1.4 1.2]{->}(60,0)(46.4,18.8)
\psline[arrowscale=1.4 1.2]{->}(80,0)(47.4,19.7)
\psline[arrowscale=1.4 1.2]{->}(20,0)(66.0,19)
\psline[arrowscale=1.4 1.2]{->}(60,0)(67,19)
\psline[arrowscale=1.4 1.2]{->}(80,0)(88,19)
\psline[arrowscale=1.4 1.2]{->}(80,0)(108,19)
\psline[arrowscale=1.4 1.2]{->}(100,0)(109,19)
\uput[0](60,0){$i$}
\uput[-10](80,0){$j$}
\uput[0](88,20){$k$}
\uput[0](100,0){$k^\prime$}
\uput[160](46,20){$\ell_1$}
\uput[0](67,20){$\ell_2$}
\uput[0](109,20){$\ell_3$}
\end{pspicture}
\end{center}
\caption{A bipartite partial order~$\pi$ is pictured, with the
ordered pairs of~$\pi$ shown as directed
edges.  (For instance, $j \prec_\pi k$ holds.)
The set $M_\pi$ is the set of minimal vertices.
The nodes $i,j,k$ shown are an example of nodes used
for a transitivity axiom
$\overline x_{i,j} \lor \overline x_{j,k} \lor \overline x_{k,i}$
of type~($\gamma$).  The nodes
$i,j,k^\prime$ are an example of the nodes
for a transitivity axiom of
type~($\beta$).}
\label{BipartiteFig}
\end{figure}

\begin{proof}
By renumbering the vertices, we can assume w.l.o.g.\ that
$M_\pi = \{0,\ldots,m{-}1\}$.
For each $k\ge m$, there is
at least one value of~$j$ such that $j \prec_\pi k$:
let $J_k$ be an arbitrary such value~$j$.  Note $J_k<m$.

Fix $i\in M_\pi$; that is, $i<m$.
Recall that the clause of type~($\alpha$) in $\GT_{\pi,n}$ for~$i$ is
$\bigvee_{j \not= i} x_{j,i}$.  We resolve this
clause successively, for each $k\ge m$ such that $i\not\prec_\pi k$,
against the clauses~$T_{i,J_k,k}$ of type~($\gamma$)
\[
     \overline x_{i,J_k}
\lor
     \overline x_{J_k,k}
\lor
     \overline x_{k,i}
\]
using resolution variables $x_{k,i}$.
(Note that $J_k\not= i$ since $i\not\prec_\pi k$.)
This yields a clause~$T^\prime_{i,m}$:
\[
\bigvee_{k\ge m \atop i \not\prec_\pi k} \overline x_{i,J_k}
~\lor~
\bigvee_{k\ge m \atop i \not\prec_\pi k} {\overline x_{J_k, k}}
~\lor~
\bigvee_{k\ge m \atop i \prec_\pi k}{x_{k,i}}
~\lor~
\bigvee_{k<m \atop k\not = i} x_{k,i}.
\]
The first two disjuncts shown above for~$T^\prime_{i,m}$
come from the side literals of the clauses~$T_{i,J_k,k}$;
the last two disjuncts come
from the literals in $\bigvee_{j \not= i} x_{j,i}$
which were not resolved on.
Since a literal $\overline x_{i,J_k}$ is the
same literal as $x_{J_k,i}$ and since $J_k<m$,
the literals in the first disjunct are
also contained in the fourth disjunct.
Thus, eliminating duplicate literals, $T^\prime_{i,m}$ is
equal to the clause
\[
\bigvee_{k\ge m \atop i \not\prec_\pi k} {\overline x_{J_k, k}}
~\lor~
\bigvee_{k\ge m \atop i \prec_\pi k}{x_{k,i}}
~\lor~
\bigvee_{k<m \atop k\not = i} x_{k,i}.
\]

Repeating this process, we obtain
derivations of the clauses~$T^\prime_{i,m}$ for all $i<m$.
The final disjuncts of these clauses,
$\bigvee_{i\not= k<m} x_{k,i}$,
are the same as the ($\alpha_\emptyset$) clauses in $\GT_m$.
Thus, the clauses~$T^\prime_{i,m}$ give
all ($\alpha_\emptyset$) clauses of~$\GT_m$, but with literals
$\overline x_{J_k, k}$ and~$x_{k,i}$ added in as side literals.
Moreover, the
clauses of type~($\beta$) in $\GT_{\pi,n}$
are exactly the transitivity clauses
of $\GT_m$.   All these clauses can be combined
exactly as in the refutation of~$\GT_m$
described in Theorem~\ref{GtProofsThm}, but carrying along
extra side literals $\overline x_{J_k,k}$ and~$x_{k,i}$,
or equivalently carrying along
literals~$\overline x_{J_k,k}$ for $J_k \prec_\pi k$,
and~$\overline x_{i,k}$ for $i \prec_\pi k$.
Since the refutation of~$\GT_m$ uses all of its transitivity
clauses and since each $\overline x_{J_k,k}$ literal is also
one of the $\overline x_{i,k}$'s,
this yields a resolution derivation~$P_\pi$
of the clause
\[
\{  \overline x_{i,k} : \hbox{$i\prec_\pi k$} \}.
\]
This is the clause $\falseBPOpi$ as desired.

Finally, we observe that $P_\pi$ is regular.  To show this,
note that the first parts of $P_\pi$ deriving the clauses~$T^\prime_{i,m}$
are regular by construction,
and they use resolution only on variables~$x_{k,i}$
with $k\ge m$, $i<m$, and $i\not\prec_\pi k$.
The remaining part of~$P_\pi$ is also regular by Theorem~\ref{GtProofsThm},
and uses resolution only on variables $x_{i,j}$ with $i,j\le m$.
\hfill $\qed$
\end{proof}

\begin{proof} of Theorem~\ref{PoolResGgtThm}.
We will show how to construct a series of ``LR partial refutations'',
denoted $R_0$, $R_1$, $R_2, \ldots$; this process
eventually terminates
with a pool resolution refutation of~$\GGT_n$.
The terminology ``LR partial'' indicates that
the refutation is being constructed in left-to-right order, with
the left part of the refutation properly formed, but with
many of the remaining leaves being labeled with bipartite partial orders
instead of with valid learned clauses or initial clauses from~$\GGT_n$.
We first describe the construction of the pool refutation, and
leave the size analysis to the end.

An LR partial refutation~$R$ is a tree with nodes
labeled with clauses that form a correct pool resolution proof,
except possibly at the leaves (the initial clauses).
Furthermore, it must satisfy the following conditions.

\begin{description}
\item[\rm a.]
$R$~is a tree. The root is labeled with the
empty clause. Each non-leaf node in~$R$ has a left child and right
child; the clause labeling the node is derived by resolution from the clauses on its two children.
\item[\rm b.] For each clause~$C$ occurring in~$R$, the
clause $C^+$ and the set of ordered
pairs $\tau(C)$ are defined by
\begin{eqnarray*}
C^+ & := & \{\overline x_{i,j} : \hbox{$\overline x_{i,j}$ is
occurs in some clause on the branch} \\
& & \quad \quad \quad \quad \quad \quad \quad \quad \quad \quad \quad
\hbox{from the root node~$R$ to~$C$}\},
\end{eqnarray*}
and $\tau(C) = \{ \langle i,j \rangle : \overline x_{i,j}\in C^+\}$.
Note that $C\subseteq C^+$ holds by definition.
In many cases,
$\tau(C)$ will be a partial specification of a partial order,
but this is not always true.  For instance, if $C$ is
a transitivity axiom, then $\tau(C)$ has a 3-cycle
and is not consistent
as a specification of a partial order.

\item[\rm d.] Each finished leaf~$L$ is
labeled with either a clause
from $\GGT_n$ or a clause that occurs to the left of~$L$
in the postorder traversal of~$R$.
\item[\rm e.] For an unfinished leaf labeled with clause~$C$,
the set $\tau(C)$ is
a partial specification of a partial order.
Furthermore, letting
$\pi$ be the bipartite
partial order associated with $\tau(C)$,
the clause $C$ is equal to~$\falseBPOpi$.
\end{description}

Property e.\ is particularly crucial and is novel to
our construction.  As shown below, each
unfinished leaf, labeled with a clause $C = \falseBPOpi$,
will be replaced by a derivation~$S$.  The derivation~$S$ often will
be based on~$P_\pi$, and thus might be expected to end with
exactly the clause~$C$; however, some of the resolution inferences
needed for $P_\pi$ might be disallowed by the pool property.
This can mean that $S$ will instead be
a derivation of a clause~$C^\prime$ such that
$C\subseteq C^\prime \subseteq C^+$.
The condition $C^\prime\subseteq C^+$ is
required because any literal $x \in C^\prime \setminus C$
will be handled by modifying the refutation~$R$
by propagating $x$ downward in~$R$
until reaching a clause that already contains~$x$.
The condition $C^\prime\subseteq C^+$ ensures that such
a clause exists.
The fact that $C^\prime\supseteq C$
will mean that enough literals are present
for the derivation to
use only (non-degenerate) resolution
inferences --- by virtue of the fact that our
constructions will pick~$C$ so that it contains
the literals that must be present for use as
resolution literals.  The extra literals
in $C^\prime \setminus C$ will be handled by propagating them down
the proof to where they are resolved on.

The construction begins by letting $R_0$ be the ``empty'' refutation,
containing just the empty clause.  Of course, this clause is
an unfinished leaf, and $\tau(\emptyset) = \emptyset$.  Thus
$R_0$ is a valid LR partial refutation.

For the induction step, $R_i$ has been constructed already.
Let $C$ be the leftmost unfinished clause in~$R_i$.
$R_{i+1}$ will be formed by replacing~$C$
by a refutation~$S$ of some clause
$C^\prime$ such that $C\subseteq C^\prime \subseteq C^+$.

We need to describe the (LR partial) refutation~$S$.  Let $\pi$ be the
bipartite partial order associated with $\tau(C)$, and consider
the derivation~$P_\pi$ from Lemma~\ref{BpoDerivationLm}.
Since $C$ is $\falseBPOpi$ by condition~e., the final line of~$P_\pi$
is the clause~$C$.  The intuition is that we would like
to let $S$ be~$P_\pi$.  The first difficulty with this is that
$P_\pi$ is dag-like, and the $LR$-refutation is intended to be
tree-like,  This difficulty, however, can be circumvented by just
expanding $P_\pi$, which is regular,
into a tree-like regular derivation with lemmas by the simple expedient of
using a depth-first traversal of~$P_\pi$.
The second, and more serious, difficulty is that
$P_\pi$ is
a derivation from~$\GT_n$, not~$\GGT_n$.  Namely,
the derivation~$P_\pi$
uses the transitivity clauses of~$\GT_n$ instead of the guarded
transitivity clauses of~$\GGT_n$.
The transitivity clauses $T_{i,j,k} :=
\overline x_{i,j}\lor \overline x_{j,k} \lor \overline x_{k,i}$
in~$P_\pi$ are handled one at a time as described below.
We will use four separate constructions:
in case~\caseiti{}, no change to~$P_\pi$ is required;
cases \caseitii{} and~\caseitiii{} require small changes;
and in the fourth case, the subproof~$P_\pi$
is abandoned in favor of ``learning'' the transitivity clause.

Before doing the four constructions, it is worth noting
that Lemma~\ref{BpoDerivationLm} implies
that no literal in~$C^+$ is used as a resolution
literal in~$P_\pi$.  To prove this, suppose $x_{i,j}$
is a resolution variable in~$P_\pi$.
Then, from Lemma~\ref{BpoDerivationLm} we have that at least one
of $i$ and~$j$ is $\pi$-minimal and that $i\not\prec_\pi j$ and
$j\not\prec_\pi i$.
Thus $i\not\prec_{\tau(C)} j$ and $j\not\prec_{\tau(C)} i$,
so $\tau(C)$ contains neither
$x_{i,j}$ nor~$\overline x_{i,j}$.

By the remark made after
Lemma~\ref{BpoDerivationLm},
no literal in~$C^+$ is used as a resolution
literal in~$P_\pi$.

\begin{description}
\item[\caseiti] If an initial transivitivity clause
of~$P_\pi$ already appears earlier in~$R_i$ (that is, to the left
of~$C$), then it is already {\em learned}, and can be used freely
in~$P_\pi$.
\end{description}
In the remaining cases \caseitii{}-\caseitiv{},
the transitivity clause~$T_{i,j,k}$ is not yet learned.
Let the guard variable for~$T_{i,j,k}$
be $x_{r,s}$, so $r=r(i,j,k)$ and $s=s(i,j,k)$.
\begin{description}
\item[\caseitii] Suppose case~\caseiti{} does not
apply and
that the guard variable~$x_{r,s}$ or its negation $\overline x_{r,s}$
is a member of~$C^+$.
The guard variable thus is used as a resolution variable somewhere
along the branch from the root to clause~$C$.
Then, as just argued above,
Lemma~\ref{BpoDerivationLm} implies that $x_{r,s}$ is not resolved on
in~$P_\pi$.
Therefore,
we can add the literal $x_{r,s}$ or~$\overline x_{r,s}$ (respectively)
to the clause $T_{i,j,k}$ and to every clause on any path below~$T_{i,j,k}$
until reaching a clause that already contains that literal.
This replaces $T_{i,j,k}$ with one of the initial clauses
$T_{i,j,k}\lor x_{r,s}$ or
$T_{i,j,k}\lor \overline x_{r,s}$ of~$\GGT_n$.
By construction, it preserves the validity of the resolution
inferences of~$R_i$ as well as the regularity
property.
Note this adds the literal
$x_{r,s}$ or $\overline x_{r,s}$ to the final
clause~$C^\prime$ of the modified~$P_\pi$.
This maintains the property that
$C\subseteq C^\prime \subseteq C^+$.
\item[\caseitiii] Suppose case~\caseiti{} does not apply and that
$x_{r,s}$ is not used as a
resolution variable anywhere below~$T_{i,j,k}$ in~$P_\pi$ and
is not a member of $C^+$.  In this case, $P_\pi$
is modified so as to derive the clause~$T_{i,j,k}$
from the two $\GGT_n$ clauses
$T_{i,j,k}\lor x_{r,s}$ and
$T_{i,j,k}\lor \overline x_{r,s}$ by resolving on~$x_{r,s}$.
This maintains the regularity of the derivation.  It also means that
henceforth $T_{i,j,k}$ will be learned.
\end{description}
If all of the transitivity clauses in~$P_\pi$ can be handled by
cases \caseiti{}-\caseitiii{},
then we use $P_\pi$ to define~$R_{i+1}$.  Namely,
let $P_\pi^\prime$ be the derivation~$P_\pi$
as modified by the applications of cases \caseitii{}
and~\caseitiii{}.
The derivation~$P_\pi^\prime$ is regular and dag-like, so we
can recast it as a tree-like derivation~$S$ with lemmas, by using
a depth-first traversal of~$P^\prime_\pi$.
The size of~$S$ is linear in the size of~$P^\prime_\pi$, since only
input lemmas need to be repeated.  The final line of~$S$
is the clause~$C^\prime$, namely $C$ plus the literals
introduced by case~\caseitii{}.
The derivation $R_{i+1}$ is
formed from~$R_i$ by replacing the clause~$C$ with the derivation~$S$
of~$C^\prime$, and then propagating each new literal $x\in C^\prime\setminus C$
down towards the root of~$R_i$, adding~$x$ to each clause below~$S$
until reaching a clause that already contains~$x$.
The derivation~$S$ contains no unfinished leaf, so $R_{i+1}$ contains
one fewer unfinished leaves than~$R_i$.

On the other hand, if even one transitivity axiom~$T_{i,j,k}$ in~$P_\pi$ is
not covered by the above three cases, then case~\caseitiv{} must be
used instead.
This introduces
a completely different construction to form~$S$:
\begin{description}
\item[\caseitiv{}] Let $T_{i,j,k}$ be any transitivity axiom in~$P_\pi$
that is not covered by cases \caseiti{}-\caseitiii{}.
In this case, the guard variable~$x_{r,s}$ is used as a resolution
variable in~$P_\pi$ somewhere below~$T_{i,j,k}$; in general,
this means we cannot use resolution on~$x_{r,s}$ to derive $T_{i,j,k}$
while maintaining the desired pool property.  Hence, $P_\pi$~is no longer used,
and
we instead will form $S$ with a short left-branching path
that
``learns'' $T_{i,j,k}$.
This will generate two or three new unfinished leaf nodes.  Since unfinished
leaf nodes in a LR partial derivation must be labeled with clauses
from bipartite partial orders, it is also necessary to attach short
derivations to these
unfinished leaf nodes to make the unfinished leaf clauses
of~$S$ correspond correctly to
bipartite partial orders.
These unfinished leaf nodes are then kept in~$R_{i+1}$
to be handled at later stages.

There are separate constructions depending on whether
$T_{i,j,k}$ is a clause of type ($\beta$) or~($\gamma$);
details are given below.
\end{description}

First suppose $T_{i,j,k}$ is of type~($\gamma$), and thus
$\overline x_{j,k}$ appears in~$C$.
(Refer to Figure~\ref{BipartiteFig}.)
Let $x_{r,s}$ be the
guard variable for the transitivity axiom~$T_{i,j,k}$.
The derivation~$S$ will have the form
\label{gammaCaseEq}
\begin{prooftree}
\AxiomC{$\overline x_{i,j},\overline x_{j,k}, \overline x_{k,i}, x_{r,s}$}
\AxiomC{$\overline x_{i,j},\overline x_{j,k}, \overline x_{k,i}, \overline x_{r,s}$}
\BinaryInfC{$\overline x_{i,j},\overline x_{j,k}, \overline x_{k,i}$}
\AxiomC{\raisebox{7pt}{$S_1$}$\proofdots$}
\kernHyps{-1ex}
\noLine
\UnaryInfC{$\overline x_{i,j},\overline x_{i,k},\overline \pi_{-[jk;jR(i)]}$}
\BinaryInfC{$\overline x_{i,j},\overline x_{j,k},\overline \pi_{-[jk;jR(i)]}$}
\AxiomC{\raisebox{7pt}{$S_2$}$\proofdots$}
\kernHyps{-1ex}
\noLine
\UnaryInfC{$\overline x_{j,i},\overline x_{j,k},\overline \pi_{-[jk;iR(j)]}$}
\BinaryInfC{$\overline x_{j,k}, \overline \pi_{-[jk]}$}
\end{prooftree}
The notation $\overline \pi_{-[jk]}$ denotes the
disjunction of the negations of the literals in
$\pi$ omitting the literal~$\overline x_{j,k}$.
We write ``$iR(j)$'' to indicate literals $x_{i,\ell}$
such that $j\prec_\pi \ell$.  (The ``$R(j)$''
means ``range of~$j$''.)  Thus $\overline \pi_{-[jk;iR(j)]}$
denotes the clause containing the negations of the literals in~$\pi$,
omitting $\overline x_{j,k}$ and any literals $\overline x_{i,\ell}$ such
that $j\prec_\pi \ell$.  The clause $\overline \pi_{-[jk;jR(i)]}$
is defined similarly, and the notation extends to more complicated
situations in the obvious way.

The upper leftmost inference of~$S$ is a resolution inference on
the variable~$x_{r,s}$.  Since $T_{i,j,k}$ is not covered by
either case \caseiti{} or~\caseitii{}, the variable~$x_{r,s}$
does not appear in or below clause~$C$ in~$R_i$.  Thus,
this use of~$x_{r,s}$ as a resolution variable does not violate
regularity.  Furthermore, since $T_{i,j,k}$ is of type~($\gamma$),
we have
$i{\not\prec_{\tau(C)}} j$,
$j{\not\prec_{\tau(C)}} i$,
$i{\not\prec_{\tau(C)}} k$, and
$k{\not\prec_{\tau(C)}} i$.
Thus
the literals $x_{i,j}$ and~$x_{i,k}$ do not appear in or below~$C$, so
they also can be resolved on without violating regularity.

Let $C_1$ and~$C_2$ be the final clauses of $S_1$ and~$S_2$, and
let $C_1^-$ be the clause below~$C_1$ and above~$C$.  The set
$\tau(C_2)$
is obtained by adding $\langle j,i \rangle$ to $\tau(C)$,
and similarly $\tau(C_1^-)$ is $\tau(C)$ plus~$\langle i,j \rangle$.
Since $T_{i,j,k}$ is type~($\gamma$), we have $i,j\in M_\pi$.
Therefore,
since $\tau(C)$ is a partial specification of a partial order,
$\tau(C_2)$ and $\tau(C_1^-)$ are also both partial specifications
of partial orders.
Let $\pi_2$ and~$\pi_1$ be the bipartite orders
associated with these two partial specifications (respectively).
We will form the subproof~$S_1$
so that it contains the clause $\falseBPOpisub 1$ as its only unfinished clause.
This will require adding inferences in~$S_1$ which
add and remove the appropriate literals.  The first step of this type
already occurs in going up from $C_1^-$ to~$C_1$ since this
has removed $\overline x_{j,k}$ and added $\overline x_{i,k}$,
reflecting the fact that $j$ is not $\pi_1$-minimal
and thus $x_{i,k}\in\pi_1$ but $x_{j,k}\notin \pi_1$.
Similarly, we will form~$S_2$ so that its only unfinished clause
is $\falseBPOpisub 2$.

We first describe the
subproof~$S_2$ of~$S$.  The situation is pictured in Figure~\ref{S2Fig},
which shows an extract from Figure~\ref{BipartiteFig}: the edges
shown in part~(a) of the figure correspond to the literals present in
the final line~$C_2$ of~$S_2$.  In particular, recall that the
literals $\overline x_{i,\ell}$ such that $j\prec_\pi \ell$ are
omitted from the last line of~$S_2$.  (Correspondingly, the edge
from $i$ to~$\ell_1$ is omitted from Figure~\ref{S2Fig}.)   The last
line~$C_2$ of~$S_2$ may not correspond to a bipartite partial order
as it may not partition $[n]$ into minimal and non-minimal elements;
thus, the last line of~$S_2$
may not qualify to be an unfinished node of~$R_{i+1}$.
(An example of this in Figure~\ref{S2Fig}(a) is
that $j\prec_{\tau(C_2)} i \prec_{\tau(C_2)}\ell_2$, corresponding to
$\overline x_{j,i}$ and~$\overline x_{i,\ell_2}$
being in the last line of~$S_2$.)
The bipartite partial
order~$\pi_2$ associated with~$\tau(C_2)$ is equal to
the bipartite partial order that agrees with~$\pi$ except that each
$i\prec_\pi \ell$ condition is replaced with
the condition $j\prec_{\pi_2} \ell$.  (This is represented in
Figure~\ref{S2Fig}(b) by the fact that
the edge from $i$ to~$\ell_2$ has been replaced by the edge
from $j$ to~$\ell_2$. Note that the vertex~$i$ is no longer a
minimal element of~$\pi_2$; that is, $i\notin M_{\pi_2}$.)
We wish to form~$S_2$ to be a regular derivation of
the clause $\overline x_{j,i},\overline\pi_{-[jk;iR(j)]}$ from
the clause $\falseBPOpisub 2$.

The subproof of~$S_2$ for replacing $\overline x_{i,\ell_2}$ in~$\overline \pi$
with $\overline x_{j,\ell_2}$ in~$\overline \pi_2$ is as follows, letting
$\overline \pi^*$ be $\overline\pi_{-[jk;iR(j);i\ell_2]}$.
\begin{equation}\label{S2FormEq}
\AxiomC{\raisebox{7pt}{$S^\prime_2$}$\proofdots$}
\kernHyps{-1ex}
\noLine
\UnaryInfC{$\overline x_{j,i},\overline x_{i,\ell_2},\overline x_{\ell_2,j}$}
\AxiomC{$\proofdots$ \raisebox{1ex}{\hbox to 0pt{rest of~$S_2$}}}
\noLine
\UnaryInfC{$\overline x_{j,k},\overline x_{j,\ell_2}, \overline x_{j,i}, \overline \pi^*$}
\BinaryInfC{$\overline x_{j,k},\overline x_{i,\ell_2},\overline x_{j,i},\overline \pi^*$}
\DisplayProof
\end{equation}
The part labeled ``rest of $S_2$'' will handle similarly the
other literals~$\ell$ such that
$i\prec_\pi \ell$ and $j\not\prec_\pi \ell$.
The final line of $S_2^\prime$ is the transitivity axiom~$T_{j,i,\ell_2}$.
This is a $\GT_n$ axiom, not a $\GGT_n$ axiom; however, it can be
handled by the methods of cases \caseiti{}-\caseitiii{}.
Namely, if $T_{j,i,\ell_2}$ has already been learned by appearing somewhere
to the left in~$R_i$, then $S^\prime_2$~is just this single clause.
Otherwise, let the guard variable for $T_{j,i,\ell_2}$ be $x_{r',s'}$.
If $x_{r',s'}$ is used as a resolution variable below~$T_{j,i,\ell_2}$, then
replace $T_{j,i,\ell_2}$ with
$T_{j,i,\ell_2}\lor x_{r',s'}$ or $T_{j,i\ell_2}\lor \overline x_{r',s'}$,
and propagate the $x_{r',s'}$ or $\overline x_{r',s'}$ to clauses down
the branch leading to~$T_{j,i,\ell_2}$ until reaching a clause that
already contains that literal.
Finally, if $x_{r',s'}$ has not been
used as a resolution variable in~$R_i$ below~$C$, then let
$S^\prime_2$~consist of a resolution inference deriving (and learning)
$T_{j,i,\ell_2}$ from
the clauses
$T_{j,i,\ell_2},x_{r',s'}$ and
$T_{j,i,\ell_2},\overline x_{r',s'}$.

To complete the construction of~$S_2$,
the inference (\ref{S2FormEq}) is repeated for each value of~$\ell$
such that $i\prec_\pi\ell$ and $j\not\prec_\pi\ell$.
The result is that $S_2$ has one unfinished leaf clause, and it
is labelled with the clause $\falseBPOpisub 2$.

\begin{figure}[t]
\psset{unit=0.06cm}     
\hfill
\begin{pspicture}(40,-10)(110,20)
\pscircle*(60,0){2pt}
\pscircle*(80,0){2pt}
\pscircle*(46,20){2pt}
\pscircle*(67,20){2pt}
\pscircle*(88,20){2pt}
\pscircle*(109,20){2pt}
\psline[arrowscale=1.4 1.2]{->}(80,0)(47.4,19.7)
\psline[arrowscale=1.4 1.2]{->}(60,0)(67,19)
\psline[arrowscale=1.4 1.2]{->}(80,0)(88,19)
\psline[arrowscale=1.4 1.2]{->}(80,0)(108,19)
\psline[arrowscale=1.4 1.2]{->}(80,0)(61,0)
\uput[180](60,0){$i$}
\uput[-10](80,0){$j$}
\uput[0](88,20){$k$}
\uput[160](46,20){$\ell_1$}
\uput[0](67,20){$\ell_2$}
\uput[0](109,20){$\ell_3$}
\rput(70,-10){(a) $\overline x_{j,k},\overline x_{i,\ell_2},\overline x_{j,i},\overline \pi^*$}
\end{pspicture}
\hfill
\begin{pspicture}(40,-10)(110,20)
\pscircle*(60,0){2pt}
\pscircle*(80,0){2pt}
\pscircle*(46,20){2pt}
\pscircle*(67,20){2pt}
\pscircle*(88,20){2pt}
\pscircle*(109,20){2pt}
\psline[arrowscale=1.4 1.2]{->}(80,0)(47.4,19.7)
\psline[arrowscale=1.4 1.2]{->}(80,0)(67,19)
\psline[arrowscale=1.4 1.2]{->}(80,0)(88,19)
\psline[arrowscale=1.4 1.2]{->}(80,0)(108,19)
\psline[arrowscale=1.4 1.2]{->}(80,0)(61,0)
\uput[180](60,0){$i$}
\uput[-10](80,0){$j$}
\uput[0](88,20){$k$}
\uput[160](46,20){$\ell_1$}
\uput[0](67,20){$\ell_2$}
\uput[0](109,20){$\ell_3$}
\rput(70,-10){(b) $\overline x_{j,k},\overline x_{i,\ell_2},\overline x_{j,i},\overline \pi^*$}
\end{pspicture}
\hfill
\caption{The partial orders for the fragment of~$S_2$ shown
in~(\ref{S2FormEq}).}
\label{S2Fig}
\end{figure}

We next describe the subproof~$S_1$ of~$S$.
The situation is shown in Figure~\ref{S1Fig}.  As in the formation
of~$S_2$, the final clause~$C_1$ in~$S_1$ may need to be modified in order
to correspond to the bipartite partial order~$\pi_1$ which
is associated with~$\tau(C_1)$.  First,
note that the literal $\overline x_{j,k}$ is already
replaced by~$\overline x_{i,k}$ in the final clause of~$S_1$.
The other change that is needed is that, for every $\ell$ such that
$j\prec_\pi \ell$ and $i\not\prec_\pi \ell$, we must
replace $\overline x_{j,\ell}$ with $\overline x_{i,\ell}$
since we have
$j\not\prec_{\pi_1} \ell$ and $i\prec_{\pi_1} \ell$.  Vertex~$\ell_3$
in Figure~\ref{S1Fig} is an example of a such a value~$\ell$.  The
ordering in the final clause of~$S_1$ is shown in part~(a), and the
desired ordered pairs of~$\pi_1$ are shown in part~(b).
Note that $j$ is no longer a minimal element in~$\pi_1$.

The replacement of $\overline x_{j,\ell_3}$
with $\overline x_{i,\ell_3}$
is effected by the following inference, letting
$\overline \pi^*$ now be
$\overline \pi_{-[jk;jR(i);j\ell_3]}$.
\begin{equation}\label{S1FormEq}
\AxiomC{\raisebox{7pt}{$S^\prime_1$}$\proofdots$}
\kernHyps{-1ex}
\noLine
\UnaryInfC{$\overline x_{i,j},\overline x_{j,\ell_3},\overline x_{\ell_3,i}$}
\AxiomC{$\proofdots$ \raisebox{1ex}{\hbox to 0pt{rest of~$S_1$}}}
\noLine
\UnaryInfC{$\overline x_{i,k},\overline x_{i,\ell_3}, \overline x_{i,j}, \overline \pi^*$}
\BinaryInfC{$\overline x_{i,k},\overline x_{j,\ell_3},\overline x_{i,j},\overline \pi^*$}
\DisplayProof
\end{equation}
The ``rest of $S_1$'' will handle similarly the
other literals~$\ell$ such that
$j\prec_\pi \ell$ and $i\not\prec_\pi \ell$.
Note that the final clause of~$S_1^\prime$ is the
transitivity axiom $T_{i,j,\ell_3}$.  The subproof $S_1^\prime$ is
formed in exactly the same way that $S_2^\prime$ was formed above. Namely,
depending on the status of the guard variable~$x_{r',s'}$
for~$T_{i,j,\ell_3}$, one of the following
is done:
\caseiti{}~the clause~$T_{i,j,\ell_3}$ is already learned
and can be used as is,
or \caseitii{}~one of $x_{r',s'}$ or~$\overline x_{r',s'}$ is
added to the clause and propagated down the proof,
or \caseitiii{}~the clause~$T_{i,j,\ell_3}$ is inferred
using resolution on~$x_{r',s'}$
and becomes learned.

To complete the construction of~$S_1$,
the inference (\ref{S1FormEq}) is repeated for each value of~$\ell$
such that $j\prec_\pi\ell$ and $i\not\prec_\pi\ell$.
The result is that $S_1$ has one unfinished leaf clause, and it corresponds to
the bipartite partial order~$\pi_1$.

\begin{figure}[t]
\psset{unit=0.06cm}     
\hfill
\begin{pspicture}(40,-10)(110,23)
\pscircle*(60,0){2pt}
\pscircle*(80,0){2pt}
\pscircle*(46,20){2pt}
\pscircle*(67,20){2pt}
\pscircle*(88,20){2pt}
\pscircle*(109,20){2pt}
\psline[arrowscale=1.4 1.2]{->}(60,0)(46.4,18.8)
\psline[arrowscale=1.4 1.2]{->}(60,0)(67,19)
\psline[arrowscale=1.4 1.2]{->}(60,0)(79,0)
\psline[arrowscale=1.4 1.2]{->}(60,0)(87.5,19)
\psline[arrowscale=1.4 1.2]{->}(80,0)(108,19)
\uput[180](60,0){$i$}
\uput[-10](80,0){$j$}
\uput[0](88,20){$k$}
\uput[160](46,20){$\ell_1$}
\uput[0](67,20){$\ell_2$}
\uput[0](109,20){$\ell_3$}
\rput(70,-10){(a) $\overline x_{i,k},\overline x_{j,\ell_3},\overline x_{i,j},\overline \pi^*$}
\end{pspicture}
\hfill
\begin{pspicture}(40,-10)(120,23)
\pscircle*(60,0){2pt}
\pscircle*(80,0){2pt}
\pscircle*(46,20){2pt}
\pscircle*(67,20){2pt}
\pscircle*(88,20){2pt}
\pscircle*(109,20){2pt}
\psline[arrowscale=1.4 1.2]{->}(60,0)(46.4,18.8)
\psline[arrowscale=1.4 1.2]{->}(60,0)(67,19)
\psline[arrowscale=1.4 1.2]{->}(60,0)(87.5,19)
\psline[arrowscale=1.4 1.2]{->}(60,0)(107.5,19)
\psline[arrowscale=1.4 1.2]{->}(60,0)(79,0)
\uput[180](60,0){$i$}
\uput[-10](80,0){$j$}
\uput[0](88,20){$k$}
\uput[160](46,20){$\ell_1$}
\uput[0](67,20){$\ell_2$}
\uput[0](109,20){$\ell_3$}
\rput(70,-10){(b) $\overline x_{i,k},\overline x_{i,\ell_3},\overline x_{i,j},\overline \pi^*$}
\end{pspicture}
\hfill
\caption{The partial orders for the fragment of~$S_1$ shown
in~(\ref{S1FormEq}).}
\label{S1Fig}
\end{figure}

That completes the construction of the subproof~$S$ for the subcase
of~\caseitiv{} where $T_{i,j,k}$ is of type~($\gamma$).  Now suppose
$T_{i,j,k}$ is of type~($\beta$).  (For instance, the values $i,j,k^\prime$
of Figure~\ref{BipartiteFig}.)
In this case the derivation~$S$ will have the form
\label{betaCaseEq}
\begin{prooftree}
\AxiomC{$T_{i,j,k}, x_{r,s}$}
\AxiomC{$T_{i,j,k}, \overline x_{r,s}$}
\BinaryInfC{$T_{i,j,k}$}
\AxiomC{\raisebox{7pt}{$S_3$}$\proofdots$}
\kernHyps{-1ex}
\noLine
\UnaryInfC{$\overline x_{i,j},\overline x_{i,k},\overline \pi_{-[jR(i),kR(i\cup j)]}$}
\BinaryInfC{$\overline x_{i,j},\overline x_{j,k},\overline \pi_{-[jR(i),kR(i\cup j)]}$}
\AxiomC{\raisebox{7pt}{$S_4$}$\proofdots$}
\kernHyps{-1ex}
\noLine
\UnaryInfC{$\overline x_{i,j},\overline x_{k,j},\overline \pi_{-[jR(i\cap k)]}$}
\BinaryInfC{$\overline x_{i,j},\overline \pi_{-[jR(i\cap k)]}$}
\AxiomC{\raisebox{7pt}{$S_5$}$\proofdots$}
\kernHyps{-1ex}
\noLine
\UnaryInfC{$\overline x_{j,i},\overline \pi_{-[iR(j)]}$}
\BinaryInfC{$\overline \pi$}
\end{prooftree}
where $x_{r,s}$ is the guard variable for~$T_{i,j,k}$.
We write $[\overline \pi_{-[jR(i\cap k)]}]$ to mean
the negations of literals in~$\pi$
omitting any literal $\overline x_{j,\ell}$ such that
both $i\prec_\pi \ell$ and $k\prec_\pi\ell$.
Similarly, $\overline \pi_{-[jR(i),kR(i\cup j)]}$ indicates
the negations of literals in~$\pi$,
omitting the literals $\overline x_{j,\ell}$
such that $i\prec_\pi \ell$ and the literals
$\overline x_{k,\ell}$ such that either
$i\prec_\pi \ell$ or $j\prec_\pi \ell$.

Note that
the resolution on~$x_{r,s}$ used to derive $T_{i,j,k}$ does
not violate regularity, since otherwise $T_{i,j,k}$ would
have been covered by case~\caseitii{}.  Likewise,
the resolutions on $x_{i,j}$, $x_{i,k}$ and $x_{j,k}$ do not
violate regularity since $T_{i,j,k}$ is
of type~($\beta$).


The subproof~$S_5$
is formed exactly like the subproof~$S_2$ above, with the exception
that now the literal $\overline x_{j,k}$ is not present.  Thus we omit
the description of~$S_5$.

We next describe the construction of the
subproof $S_4$.  Let $C_4$ be the final clause of~$S_4$;
it is easy to check that $\tau(C_4)$ is a partial specification
of a partial order.
As before, we must derive~$C_4$
from the clause $\falseBPOpisub 4$ where $\pi_4$ is the bipartite
partial order associated with the partial order~$\tau(C_4)$.
A typical situation is shown in Figure~\ref{S4Fig}.  As pictured
there, it is necessary to add the
literals $\overline x_{i,\ell}$ such that $j\prec_\pi \ell$ and
$i\not\prec_\pi \ell$, while removing $\overline x_{j,\ell}$; examples
of this are $\ell$ equal to $\ell_2$ and~$\ell_3$
in Figure~\ref{S4Fig}.
At the same time, we must add the
literals $\overline x_{k,\ell}$ such that $j\prec_\pi \ell$ and
$k\not\prec_\pi \ell$, while removing $\overline x_{j,\ell}$; examples
of this are $\ell$ equal to $\ell_1$ and~$\ell_2$
in the same figure.

For a vertex~$\ell_3$ such that
$j\prec_\pi \ell_3$ and
$k\prec_\pi \ell_3$ but
$i\not\prec_\pi \ell_3$, this is done similarly to the inferences
(\ref{S2FormEq}) and~(\ref{S1FormEq})
but without the side literal~$\overline x_{j,k}$:
\begin{equation}\label{S4FormIonlyEq}
\AxiomC{\raisebox{7pt}{$S^\prime_4$}$\proofdots$}
\kernHyps{-1ex}
\noLine
\UnaryInfC{$\overline x_{i,j},\overline x_{j,\ell_3},\overline x_{\ell_3,i}$}
\AxiomC{$\proofdots$ \raisebox{1ex}{\hbox to 0pt{rest of~$S_4$}}}
\noLine
\UnaryInfC{$\overline x_{i,\ell_3},\overline x_{k,j}, \overline x_{i,j}, \overline \pi^*$}
\BinaryInfC{$\overline x_{j,\ell_3},\overline x_{k,j},\overline x_{i,j},\overline \pi^*$}
\DisplayProof
\end{equation}
Here $\overline \pi^*$ is
$\overline \pi_{-[jR(i\cap k);j\ell_3]}$.
The transitivity axiom~$T_{i,j,\ell_3}$ shown as the last line
of~$S_4^\prime$ is handled exactly as before.  This construction is
repeated for all such~$\ell_3$'s.

The vertices~$\ell_1$ such that
$j\prec_\pi \ell_1$ and
$i\prec_\pi \ell_1$ but
$k\not\prec_\pi \ell_1$ are handled in exactly the same way.
(The side literals of~$\pi^*$ change each time to reflect the
literals that have already been replaced.)

Finally, consider a vertex~$\ell_2$ such that
$i\not\prec_\pi \ell_2$ and
$j\prec_\pi \ell_2$ and
$k\not\prec_\pi \ell_2$.  This is handled by the derivation
\begin{prooftree}
\AxiomC{\raisebox{7pt}{$S^{\prime\prime}_4$}$\proofdots$}
\kernHyps{-1ex}
\noLine
\UnaryInfC{$\overline x_{i,j},\overline x_{j,\ell_2},\overline x_{\ell_2,i}$}
\AxiomC{\raisebox{7pt}{$S^{\prime\prime\prime}_4$}$\proofdots$}
\kernHyps{-1ex}
\noLine
\UnaryInfC{$\overline x_{k,j},\overline x_{j,\ell_2},\overline x_{\ell_2,k}$}
\AxiomC{$\proofdots$ \raisebox{1ex}{\hbox to 0pt{rest of~$S_4$}}}
\noLine
\UnaryInfC{$\overline x_{i,j},\overline x_{i,\ell_2},\overline x_{k,j},\overline x_{k,\ell_2},\overline \pi^*$}
\BinaryInfC{$\overline x_{i,j},\overline x_{i,\ell_2},\overline x_{k,j},\overline x_{j,\ell_2},\overline \pi^*$}
\BinaryInfC{$\overline x_{i,j},\overline x_{k,j},\overline x_{j,\ell_2},\overline \pi^*$}
\end{prooftree}
As before, the set~$\pi^*$ of side literals is changed to reflect
the literals that have already been added and removed as~$S_4$ is being
created.
The subproofs $S^{\prime\prime}_4$ and~$S^{\prime\prime\prime}_4$
of the transitivity axioms $T_{i,j,\ell_2}$ and~$T_{k,j,\ell_2}$
are handled exactly as before, depending on the status of their
guard variables.

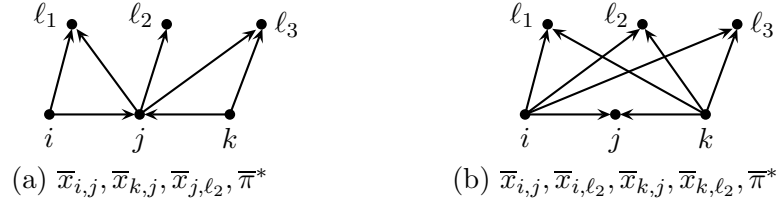
\begin{figure}[t]
\psset{unit=0.06cm}     
\hfill
\begin{pspicture}(16,-15)(72,24)
\pscircle*(20,0){2pt}
\pscircle*(40,0){2pt}
\pscircle*(60,0){2pt}
\pscircle*(25,20){2pt}
\pscircle*(46,20){2pt}
\pscircle*(67,20){2pt}
\psline[arrowscale=1.4 1.2]{->}(20,0)(25,19)
\psline[arrowscale=1.4 1.2]{->}(40,0)(26,19)
\psline[arrowscale=1.4 1.2]{->}(40,0)(46,19)
\psline[arrowscale=1.4 1.2]{->}(40,0)(65.8,18.8)
\psline[arrowscale=1.4 1.2]{->}(60,0)(67.2,19)
\psline[arrowscale=1.4 1.2]{->}(20,0)(39,0)
\psline[arrowscale=1.4 1.2]{->}(60,0)(41,0)
\uput[-90](20,0){$i$}
\uput[-90](40,0){$j$}
\uput[-90](60,0){$k$}
\uput[160](25,20){$\ell_1$}
\uput[160](46,20){$\ell_2$}
\uput[0](67,20){$\ell_3$}
\rput(40,-15){(a) $\overline x_{i,j},\overline x_{k,j},\overline x_{j,\ell_2},\overline \pi^*$}
\end{pspicture}
\hfill
\begin{pspicture}(16,-15)(72,24)
\pscircle*(20,0){2pt}
\pscircle*(40,0){2pt}
\pscircle*(60,0){2pt}
\pscircle*(25,20){2pt}
\pscircle*(46,20){2pt}
\pscircle*(67,20){2pt}
\psline[arrowscale=1.4 1.2]{->}(20,0)(25,19)
\psline[arrowscale=1.4 1.2]{->}(20,0)(45.5,19)
\psline[arrowscale=1.4 1.2]{->}(20,0)(66.0,19)
\psline[arrowscale=1.4 1.2]{->}(60,0)(26,19)
\psline[arrowscale=1.4 1.2]{->}(60,0)(46.5,19)
\psline[arrowscale=1.4 1.2]{->}(60,0)(67,19)
\psline[arrowscale=1.4 1.2]{->}(20,0)(39,0)
\psline[arrowscale=1.4 1.2]{->}(60,0)(41,0)
\uput[-90](20,0){$i$}
\uput[-90](40,0){$j$}
\uput[-90](60,0){$k$}
\uput[160](25,20){$\ell_1$}
\uput[160](46,20){$\ell_2$}
\uput[0](67,20){$\ell_3$}
\rput(40,-15){(b) $\overline x_{i,j},\overline x_{i,\ell_2},\overline x_{k,j},\overline x_{k,\ell_2},\overline \pi^*$}
\end{pspicture}
\hfill
\caption{The partial orders as changed by~$S_4$.}
\label{S4Fig}
\end{figure}

Finally, we describe how to form the subproof~$S_3$.  For this, we
must form the bipartite partial order~$\pi_3$ which associated with
the partial order~$\tau(C_3)$,
where $C_3$ is the final clause of~$S_3$.
To obtain $\overline\pi_3$, we need to add
the literals $\overline x_{i,\ell}$ such that
$i\not\prec_\pi \ell$ and such that
either $ j\prec_\pi \ell$ or $k\prec_\pi \ell$, while removing
any literals $\overline x_{j,\ell}$ and $\overline x_{k,\ell}$.
This is done by exactly the same
construction used above in~(\ref{S4FormIonlyEq}).  The literals
in $\overline \pi_{-[jR(i);kR(i\cup j)]}$ are exactly the literals needed
to carry this out.  The construction is quite similar to the above
constructions, and we omit any further description.

That completes the description of how to
construct the LR partial refutations~$R_{i}$.
The process stops once some~$R_i$ has no unfinished clauses.
We claim that the process stops after polynomially many stages.

To prove this, recall that $R_{i+1}$ is formed
by handling the leftmost unfinished clause using
one of cases \caseiti{}-\caseitiv{}.  In the first three cases,
the unfinished clause is replaced by a derivation based on $P_\pi$
for some bipartite order~$\pi$.  Since $P_\pi$ has size $O(n^3)$, this
means that the number of clauses in~$R_{i+1}$ is at most the
number of clauses in~$R_i$ plus $O(n^3)$.  Also, by construction, $R_{i+1}$~has
one fewer unfinished clauses than~$R_i$.  In case~\caseitiv{} however,
$R_{i+1}$ is formed by adding up to $O(n)$ many clauses to~$R_i$
plus adding either two or three new unfinished leaf clauses.  In addition,
case~\caseitiv{} always causes at least one transitivity axiom~$T_{i,j,k}$
to be learned.
Therefore, case~\caseitiv{} can occur at most $2{n \choose 3} = O(n^3)$ times.
Consequently at most $3\cdot 2 {n \choose 3} = O(n^3)$ many
unfinished clauses are added throughout the entire process.
It follows that the process stops with~$R_i$ having no unfinished clauses
for some $i\le 6 {n \choose 3}=O(n^3)$.  Therefore there is a pool
refutation of $\GGT_n$ with $O(n^6)$ lines.

By inspection, each clause in the refutation contains $O(n^2)$
literals.  This is because the largest clauses are those corresponding
to (small modifications of) bipartite partial orders, and because
bipartite partial orders can contain at most $O(n^2)$ many ordered pairs.
Furthermore, the refutations~$P_n$ for the graph tautology $\GT_n$ contain
only clauses of size $O(n^2)$.
\hfill \\
Q.E.D. Theorem~\ref{PoolResGgtThm} \hfill $\qed$
\end{proof}

Theorem~\ref{regRtiGgtThm} is proved with nearly the same construction.  In fact, the only change needed for
the proof is the construction of~$S$ from~$P_\pi^\prime$.
Recall that in the proof of Theorem~\ref{PoolResGgtThm}, the
pool derivation~$S$ was formed
by using a depth-first traversal of~$P$.  This is not
sufficient for
Theorem~\ref{regRtiGgtThm}, since now the derivation~$S$ must use only input
lemmas.  Instead, we use Theorem~3.3 of~\cite{BHJ:ResTreeLearning}, which
states that a (regular) dag-like resolution derivation can
be transformed into a (regular) tree-like derivation with input lemmas.
Forming $S$ in this way from~$P^\prime_\pi$ suffices for the
proof of Theorem~\ref{regRtiGgtThm}: the lemmas of~$S$ are either
transitive closure axioms derived earlier in~$R_i$ or are
derived by input subproofs earlier in the post-ordering of~$S$.
Since the transitive closure axioms that appeared
earlier in~$R_i$ were derived by resolving two $\GGT_n$ axioms,
the lemmas used in~$S$ are all input lemmas.

The transformation of Theorem~3.3 of~\cite{BHJ:ResTreeLearning}
may multiply the size of the derivation
by the depth of the original derivation.  Since
it is possible to form the proofs~$P_\pi$
with depth $O(n)$, the overall size of the
pool resolution refutations with input lemmas is $O(n^7)$.
This completes the proof of Theorem~\ref{regRtiGgtThm}.
\hfill $\qed$

\section{Greedy, unit-propagating DPLL with clause learning}\label{GreedySect}

This section discusses how the refutations in
Theorems \ref{PoolResGgtThm} and~\ref{regRtiGgtThm}
can be modified so as to ensure
that the refutations are greedy and unit-propagating.

\begin{definition}
Let $R$ be a
tree-like regular w-resolution refutation
with input lemmas.
For $C$ a clause in~$R$, let $C^+$ be the set
of literals which occur as literals or phantom literals
in clauses on the path from~$C$ to the root of~$R$.  (Recall that
``phantom literals'' are literals used for w-resolution
that are not actually present in the clauses.)
Also, let $\Gamma(C)$ be
the set of clauses of $\Gamma$
plus every clause $D <_R C$ in~$R$ that has
been derived by an input subproof
and thus is available as a learned
clause to aid in the derivation of~$C$.

The refutation $R$ is {\em greedy and unit-propagating}
provided
that, for each clause~$C$ of~$R$, if
there is an input derivation
from~$\Gamma(C)$ of some clause $C^\prime\subseteq C^+$
which does not resolve on any literal in~$C^+$,
then $C$~is derived in~$R$ by such a
derivation.
\end{definition}

Note that, as proved in~\cite{BKS:clauselearning},
the condition that there is a input derivation from $\Gamma(C)$ of
some $C^\prime \subseteq C^+$ which does not resolve
on $C^+$ literals
is equivalent to the condition that if all literals of $C^+$
are set false
then unit propagation yields a contradiction from $\Gamma(C)$.
(In \cite{BKS:clauselearning}, these are called ``trivial'' proofs.)
This justifies the terminology ``unit-propagating''.

The definition of ``greedy and unit-propagating''
is actually a bit more restrictive than
necessary, since DPLL algorithms may actually learn multiple
clauses at once, and this can mean that $C$ is not derived from
a single input proof but rather from a combination of several input proofs
as described in the proof of Theorem 5.1 in~\cite{BHJ:ResTreeLearning}.

\begin{theorem}\label{GreedyRegRtiThm}
The guarded graph tautology principles $\GGT_n$ have greedy, unit-propagating,
polynomial size, tree-like, regular w-resolution refutations with input lemmas.
\end{theorem}

\begin{proof}
We indicate how to modify the proofs
of Theorems \ref{PoolResGgtThm} and~\ref{regRtiGgtThm}.
We again build tree-like LR partial
refutations satisfying the same
properties a.-e.\ as before, except now w-resolution inferences
are permitted.  Instead of
being formed in distinct stages $R_0, R_1, R_2,\ldots$,
the w-resolution refutation~$R$ is constructed by one continuing process.
This construction incorporates all
of transformations \caseiti{}-\caseitiv{} and also incorporates
the construction of Theorem 3.3 of~\cite{BHJ:ResTreeLearning}.

At each point in the construction, we will be scanning
the so-far constructed partial w-resolution refutation~$R$
in preorder, namely in
depth-first, left-to-right order.  That is to say,
the construction recursively processes a node in the proof tree,
then its left subtree,
and then its right subtree.  During most steps of
the preorder scan,
the partial refutation~$R$ is modified by changing
the part that comes subsequently in the preorder,
but the construction may also add and remove literals
from clauses below the current clause~$C$.
When the preorder scan reaches a clause~$C$ that
has an input derivation~$R^\prime$ from $\Gamma(C)$
of some $C^\prime\subseteq C$
that does not
resolve on~$C^+$, then some such~$R^\prime$ is inserted
into~$R$ at that point.
When the preorder scan reaches an unfinished leaf~$C=C_0$,
then a
(possibly exponentially large) derivation~$P_\pi^*$ is added
as its derivation.  The construction continues processing~$R$
by scanning~$P_\pi^*$ in preorder, with the end result
that either (1)~$P_\pi^*$ is succcessfully processed and reduced
to only polynomial size or (2)~the preorder scan of~$P_\pi^*$
reaches a transitivity clause~$T_{i,j,k}$ of the type that
triggered case~\caseitiv{} of Theorem~\ref{PoolResGgtThm}.
In the latter case,
the preorder scan backs up to the root clause~$C_0$ of~$P_\pi^*$,
replaces $P_\pi^*$ with the
derivation~$S$ constructed in case~\caseitiv{}
of Theorem~\ref{PoolResGgtThm},
and restarts the preorder scan at clause~$C_0$.

We describe the actions of the preorder scan in more
detail. Initially, $R$ is the ``empty''
derivation, with the empty clause as its only
(unfinished) clause.  A clause~$C$ encountered during
the preorder scan of~$R$ is handled by one of
the following.

\begin{description}
\item[\caseitip] Suppose that some $C^\prime\subseteq C^+$
can be derived by an input
derivation from~$\Gamma(C)$ that does not resolve on any literals
of~$C^+$.
Fix any such~$C^\prime\subseteq C^+$,
and replace the subderivation
in~$R$ of the clause~$C$ with such a derivation
of~$C^\prime$ from $\Gamma(C)$.
Any extra literals in $C^\prime\setminus C$
are in $C^+$ and are propagated down
until reaching a clause where
they already appear, or occur as a phantom
literal.  There may also be
literals in $C \setminus C^\prime$: these
literals are removed as necessary from clauses
below~$C^\prime$ in~$R$ to maintain the property
of~$R$ containing correct w-resolution inferences.
Note that this can convert resolution inferences into
w-resolution inferences.

The clause $C^\prime$ is now a learned clause.
Note that this case includes
transitivity clauses $C=C^\prime=T_{i,j,k}$
that satisfy the conditions of cases \caseiti{}-\caseitiii{}
of Theorem~\ref{PoolResGgtThm}
\item[\caseitiip] If case~\caseitip{} does not apply,
and $C$ is not a leaf node, then $R$ is unchanged
at this point and the depth-first traversal proceeds
to the next clause.
\item[\caseitiiip] If $C$ is an unfinished clause
of the form $\falseBPOpi$, let $P_\pi$ be
as before.
Recall that no literal in~$C^+$ is resolved on in~$P_\pi$.
Unwind the proof~$P_\pi$
into a tree-like regular refutation~$P_\pi^*$
that is possibly exponentially big, and attach
$P_\pi^*$ to~$R$ as a proof of~$C$.  Mark the
position of~$C$ by setting $C_0=C$
in case it is necessary to
later backtrack to~$C$.  Then continue the preorder
scan by traversing into~$P_\pi^*$.

\item[\caseitivp] Otherwise, $C$ is an initial clause of
the form $T_{i,j,k}$ and since case~\caseitip{} does not apply,
one of $T_{i,j,k}$'s guard literals~$x$,
namely $x=x_{r,s}$ or $x=\overline x_{r,s}$,
is in~$C^+$.  If $C$ is {\em not} inside
the most recently added~$P_\pi^*$ or if $x\in C_0^+$,
then replace
$T_{i,j,k}$ with $T_{i,j,k}\lor x$,
and propagate the literal~$x$
downward in the refutation until reaching a
clause where it appears as a literal or a phantom literal.
Otherwise, the preorder scan backtracks to the
root clause~$C_0$ of $P_\pi^*$, and replaces $P_\pi^*$ with
the partial resolution refutation~$S$ formed in
case~\caseitiv{} of Theorem~\ref{PoolResGgtThm}.
\end{description}

It is clear that this process eventually halts with
a valid greedy, unit-propagating, tree-like w-resolution
refutation.  We claim that it also yields
a polynomial size refutation.  Consider
what happens when a derivation~$P_\pi^*$ is inserted.
If case~\caseitivp{} is triggered, then the proof~$S$ is
inserted in place of~$P_\pi^*$, so the size
of~$P_\pi^*$ does not matter.
If case~\caseitivp{} is not triggered,
then, as in the proof of Theorem 3.3 of~\cite{BHJ:ResTreeLearning},
the preorder scan of~$P_\pi^*$ modifies (the possibly
exponentially large) $P_\pi^*$ to
have polynomial size.
Indeed, as argued in~\cite{BHJ:ResTreeLearning},
any clause~$C$
in~$P_\pi^*$ will occur at most $d_C$~times in the modified
version of~$P_\pi^*$ where $d_C$ is the depth of the derivation
of~$C$ in the original~$P_\pi$.  This is because, $C$ will
have been learned by an input derivation once it has appeared no more than~$d_C$
times in the modified derivation~$P_\pi^*$.  This is
proved by induction on~$d_C$.

Consider the situation where $S$ has
just been inserted in place
of~$P_\pi^*$ in case~\caseitiiip{}.
The transitivity clause~$T_{i,j,k}$
is not yet learned at this point,
since otherwise case~\caseitip{} would have
applied.  We claim, however, that $T_{i,j,k}$
is learned as~$S$ is traversed.  To prove this,
since $T_{i,j,k}$ is manifestly derived
by an input derivation and since its guard
literals $x_{r,s}$ and~$\overline x_{r,s}$ do
not appear in~$C_0^+$, it is
enough to show that the clause~$T_{i,j,k}$ is
reached in the preorder traversal scan of~$S$.
This, however, is an immediate consequence of the fact that
$T_{i,j,k}$ was reached in the preorder scan of~$P_\pi^*$
and triggered case~\caseitivp{},
since if case~\caseitip{} applies to $T_{i,j,k}$
or to any clause below $T_{i,j,k}$
in the preorder scan of~$S$,
then it certainly also applies $T_{i,j,k}$ or
some clause below $T_{i,j,k}$
in the preorder scan of~$P_\pi^*$.

The size of the final refutation~$R$ is bounded
the same way as in the proof of Theorem~\ref{regRtiGgtThm},
and this completes the proof of Theorem~\ref{GreedyRegRtiThm}.
\hfill $\qed$
\end{proof}

\begin{theorem}\label{DPLLpolySizeThm}
There are DPLL search procedures with clause
learning which are greedy, unit-propagating, but do not use
restarts, that
refute the $\GGT_n$ clauses in polynomial time.
\end{theorem}

We give a sketch of the
proof.
The construction for the proof of Theorem~\ref{GreedyRegRtiThm}
requires only that the clauses~$T_{i,j,k}$ are learned
whenever possible, and does not depend on whether any other
clauses are learned.
This means that the following
algorithm for DPLL search with clause learning
will always succeed in finding a refutation of
the $\GGT_n$ clauses: At each point, there is
a partial assignment~$\tau$.
The search algorithm must do one of the following:
\setlength{\parsep}{0in}
\setlength{\itemsep}{0in}
\begin{description}
\item[\rm (1)] If unit propagation yields a
contradiction, then learn a clause $T_{i,j,k}$ if
possible, and backtrack.
\item[\rm (2)] Otherwise, if there are
any literals in the transitive closure of
the bipartite partial order
associated with~$\tau$ which are not assigned
a value,
branch on one of these literals to set
its value.  (One of the true or false assignments
yields an immediate
conflict, and may allow learning a clause~$T_{i,j,k}$.)
\item[\rm (3)] Otherwise,
determine whether there is a clause~$T_{i,j,k}$
which is used in the proof~$P_\pi$ whose guard literals
are resolved on in~$P_\pi$.  (See Lemma~\ref{BpoDerivationLm}.)
If not, do a DPLL traversal of~$P_\pi$,
eventually backtracking from the assignment~$\tau$.
\item[\rm (4)] Otherwise, the clause
$T_{i,j,k}$ blocks
$P_\pi$ from being traversed in polynomial time.  Branch on
its variables in the order given
in the proof of Theorem~\ref{PoolResGgtThm}.
From this, learn the clause
$T_{i,j,k}$.
\end{description}

\bibliographystyle{siam}
\bibliography{logic,cstheory}

\end{document}